\documentclass[10pt,a4paper]{article}
\usepackage{graphicx}
\usepackage{amssymb,amsmath,amsthm}
\newcommand{\e}{\mathrm{e}}
\newcommand{\D}{\mathrm{d}}
\newcommand{\C}{\mathbb{C}}
\newcommand{\N}{\mathbb{N}}
\newcommand{\R}{\mathbb{R}}

\newcommand{\Oo}{\mathcal{O}}

\newcommand{\Hm}[1]{\leavevmode{\marginpar{\tiny%
$\hbox to 0mm{\hspace*{-0.5mm}$\leftarrow$\hss}%
\vcenter{\vrule depth 0.1mm height 0.1mm width \the\marginparwidth}%
\hbox to
0mm{\hss$\rightarrow$\hspace*{-0.5mm}}$\\\relax\raggedright #1}}}
\newtheorem{claim}{Claim}[section]
\newtheorem{theorem}[claim]{Theorem}
\newtheorem{lemma}[claim]{Lemma}

\theoremstyle{definition}
\newtheorem{remark}[claim]{Remark}

\begin{document}

\title{Strong coupling asymptotics for Schr\"odinger operators with an interaction
 supported by an open arc in three dimensions}

%
\author{  Pavel Exner$\,^a$\thanks{The research has been supported by the project \mbox{14-06818S} of the Czech Science Foundation and by the project DEC-2013/11/B/ST1/03067 of the Polish  National Science Centre.} \,\,  and
Sylwia Kondej$\,^b$}
%
\date{
\small \emph{
\begin{quote}
\begin{itemize}
\item[$a)$]
Department of Theoretical Physics, Nuclear Physics Institute,
Czech Academy of Sciences, Hlavn\'{\i} 130, 25068 \v{R}e\v{z} near
Prague  \\ Doppler Institute, Czech Technical University,
B\v{r}ehov\'{a} 7, 11519 Prague, Czech Republic; exner@ujf.cas.cz
\item[$b)$]
Institute of Physics, University of Zielona G\'ora, ul.\ Szafrana
4a, 65246 Zielona G\'ora, Poland; skondej@if.uz.zgora.pl
\end{itemize}
\end{quote}}}
%


\maketitle
\begin{abstract}
\noindent We consider Schr\"odinger operators with a strongly
attractive singular interaction supported by a finite curve
$\Gamma$ of lenghth $L$ in $\R^3$. We show that if $\Gamma$ is
$C^4$-smooth and has regular endpoints, the $j$-th eigenvalue of
such an operator has the asymptotic expansion $\lambda_j
(H_{\alpha,\Gamma})= \xi_\alpha +\lambda
_j(S)+\mathcal{O}(\mathrm{e}^{\pi \alpha })$ as the coupling
parameter $\alpha\to\infty$, where $\xi_\alpha =
-4\,\mathrm{e}^{2(-2\pi\alpha +\psi(1))}$ and $\lambda _j(S)$ is
the $j$-th eigenvalue of the Schr\"odinger operator
$S=-\frac{\D^2}{\D s^2 }- \frac14 \gamma^2(s)$ on $L^2(0,L)$ with
Dirichlet condition at the interval endpoints in which $\gamma$ is
the curvature of $\Gamma$.
\end{abstract}

\bigskip



\section{Introduction}\label{introduction}

Schr\"odinger operators with singular interactions supported by
zero-measure subsets of the configuration space attracted
attention of mathematicians already several decades ago. One of
the reasons was that  their spectral analysis can be often done
explicitly to a degree. The simplest situation the interaction
support is a discrete set of points has been studied thoroughly,
see the monograph \cite{AGHH}. Later singular interactions
supported by manifolds of codimension one were analyzed
\cite{BEKS, BT}. From 2001 one witnessed a new wave of interest to
such operators with attractive interactions. It was motivated by
two facts. On the one  hand such operators appeared to be good
models for a number of tiny structures studied in solid state
physics, and on the other hand, an intriguing connection between
spectral properties of such operators and the geometry of the
interaction support was found in \cite{EI}. The most prominent manifestation of this connection is the existence of purely geometrically induced bound states \cite{EI, EK4}; a review of the work done
in this area can be found in the paper \cite{Ex08}.

A question of a particular importance concerns the strong coupling
behavior of the spectra of such operators. In this asymptotic
regime the eigenfunctions are strongly localized around the
interaction support and one expects an effective lower-dimensional
dynamics to play role. The corresponding asymptotic expansion were
demonstrated in several situations, for curves in $\R^2$
\cite{EY1, EY2} and $\R^3$ \cite{EK4} as well as for surfaces in
$\R^3$ \cite{EK1}. In all those cases, the next to leading term
was governed by a Schr\"odinger operator of the dimension of the
interaction support with an effective, geometrically induced
potential.

The technique used in all those papers was a combination of
bracketing estimates with suitable coordinate transformations
which allowed one to translate the geometry of the problem into
coefficients of the comparison operator. It had a serious
restriction as it required that the manifold supporting the
interaction has no boundary, being either infinite or a closed
curve or surface. Manifolds with a boundary have been also
considered but only in situations when the latter is connected
with a shrinking `hole' in a surface \cite{EY3} or a shrinking
hiatus in a curve \cite{EK4}. The methods used in those cases were
perturbative and did not help to address the problem of strong
coupling asymptotics for a fixed manifolds with a boundary.

A way to overcome the difficulties with the boundary was proposed
in \cite{EP13}. It used a bracketing estimate again, this time in
the neighborhood of an extended curve, together with a suitable
integral representation of the eigenfunctions. In this way
two-dimensional Schr\"odinger operators with an interaction
supported an open arc, i.e. a finite non-closed curve in $\R^2$,
were treated in \cite{EP13}. It was shown that next-to-leading
term is again given by an auxiliary one-dimensional Schr\"odinger
operator with the curvature-induced potential, this time with the
Dirichlet conditions at the endpoints of the interval that
parametrizes the curve. Our aim in the present paper is to analyze
the analogous problem for Schr\"odinger operators with interaction
support of codimension two being a finite non-closed curve in
$\R^3$.

Such an extension is no way trivial, in particular, due to a
different and more singular character of the interaction. To be
specific, we consider a non-relativistic spinless particle exposed
to a singular interaction supported by a finite curve $\Gamma
\subset \R^3$ with `free' ends. In the following section we shall
describe how one can construct Hamiltonian of such  a system, in
brief it will be identified with a self-adjoint extension of
$-\dot{\Delta } = -\Delta \upharpoonright _{C^{\infty}_{0} (\R^3
\setminus \Gamma )}$, where the latter denotes the restriction of
the Laplacian $-\Delta\,:\, W^{2,2}(\R^3)\to L^2 (\R^3)$ to the
set $C^{\infty}_{0}(\R^3 \setminus\Gamma)$. The self-adjoint
extensions are determined by means of boundary conditions imposed
at $\Gamma $ and classified by a parameter $\alpha\in\R$ which can
be regarded coupling constant\footnote{A caution is needed,
though, due the  particular character of singular interactions
with support of codimension larger than one. In particular, it is
better to avoid formal expressions of the type
$-\Delta-\alpha\delta(x-\Gamma)$, because the limits $\alpha\to
\pm\infty$ for $H_{\alpha,\Gamma}$ correspond to absence of the
interaction and the strong attraction asymptotics, respectively.}.
We denote those operators as $H_{\alpha,\Gamma}$.

We are going to find the \emph{asymptotics of eigenvalues of
$H_{\alpha,\Gamma}$ in the regime of strong coupling, $\alpha\to
-\infty$.} As in the other cases mentioned above the expansion
starts with a divergent term. We are interested in the next one,
expected to be a one-dimensional Schr\"odinger operator with the
same symbol as in the case when $\Gamma$ is a loop. When the
eigenfunctions are strongly localized aroud $\Gamma$ one may expect
their rapid falloff not only transversally but also with the
distance from the curve ends. This suggests that the effective
dynamics should involve Dirichlet boundary conditions as in the
case of codimension one. We are going to show that under mild
regularity assumptions it is indeed the case: the $j$th eigenvalue
of $H_{\alpha,\Gamma}$ admits the expansion
$$
\lambda_j (H_{\alpha,\Gamma})= \xi_\alpha +\lambda_j
(S)+\mathcal{O}(\mathrm{e}^{\pi \alpha})\quad \mathrm{as}\quad
\alpha \to -\infty\,,
$$
where $\lambda_j(S)$ stands for the $j$th eigenvalue of
$$
S=-\frac{\D^2}{\D s^2 }- \frac14 \gamma^2(s)\,:\, D(S)\rightarrow
L^2 (0, L)\,,
$$
with $D(S):= W^{1,2}_0 (0,L)\cap W^{2,2} (0,L)$, where $L$ and
$\gamma$ are the length of $\Gamma$ and its signed curvature,
respectively, and $\xi_\alpha$ is given by \eqref{2D ev} below.

The result will be stated properly together with the outline of
the proof in Sec.~\ref{s: main}, cf. Theorem~\ref{th-main}. Before
that we collect in the next section the needed preliminaries,
Secs.~\ref{s: est W} and \ref{s: ev-as} are devoted to completion
of the proof.

\section{Preliminaries}  \label{sec-prelim}

\subsection{Strongly singular interactions}

As we have mentioned the character of interactions with support of
codimension two is different and more singular than in the case of
codimension one. Let us first recall well known facts about point
interactions in dimension two which illuminate how our
curve-supported potential behaves in the transverse plane to
$\Gamma$. Consider a single point interactions placed at $y\in
\R^2$. The corresponding Hamiltonian is constructed as a
self-adjoint extension of the symmetric operator
$-\dot\Delta\,:=\,-\Delta \upharpoonright
_{C_0^\infty(\R^2\setminus \{y\})}$, i.e. the restriction of
$-\Delta\, :\, W^{2,2}(\R^2) \to L^2 (\R^2)$ to the set
$C_0^\infty (\R^2 \setminus \{y\})$. Functions from the domain of
the adjoint of $-\dot\Delta$ admit a logarithmic singularity at
the point $y$, in its vicinity they can be written as $f(x) =
-\Xi(f) \ln|x-y| + \Omega(f) + \Oo(|x-y|)$. Self-adjoint
extensions are then characterized by a parameter
$\alpha\in\R\cup\{\infty\}$ being characterized by the boundary
condition
 \begin{equation} \label{2Dbc}
2\pi \alpha \Xi(f)= \Omega(f)\,,
 \end{equation}
 which in the case $\alpha=\infty$ is a just shorthand for $\Xi(f)=0$. With the exception of this case, each extension has a single negative eigenvalue equal to
 \begin{equation} \label{2D ev}
\xi_\alpha = -4\,\mathrm{e}^{2(-2\pi\alpha +\psi(1))}\,,
 \end{equation}
where $\psi$ is the digamma function. In the following we will use
notation $f\in \mathrm{bc}(\alpha , \Gamma )$ for a a function
$f\in L^{2}(\R^2)$ satisfying (\ref{2Dbc}). We refer to
\cite[Chap.~I.5]{AGHH} for these and other facts concerning
two-dimensional point interactions.

\subsection{Geometry of the potential support and its neighborhood}

\emph{Geometry of $\Gamma.\:$} Let $\Gamma $ be a finite
non-closed $C^4 $ smooth curve in $\R^3$ of the length $L$. In
addition, we suppose that $\Gamma $ has no self-intersections.
Without loss of generality we may assume  that it is parameterized
by its arc length, and  we keep the notation $\Gamma \,:\,
I\rightarrow \R^3$, $\:I :=(0\,,L)$, for the corresponding
function. Furthermore, we assume that the curve has \emph{regular
ends}, i.e. there exists $d_0\geq 0$ such that for any $d\in
[0\,,d_0]$ the curve $\Gamma $ admits a regular extension
$\Gamma^\mathrm{ex}_d$. By this we mean that
$\Gamma^\mathrm{ex}_d$ is the graph of a $C^4 $ smooth function
$\Gamma^\mathrm{ex}_d\,:\, I_d \to \R^3$ with $I_d :=(-d\,,L+d)$
and the restriction of $\Gamma^\mathrm{ex}_d$ to $I$ coincides
with the original curve, in other words, $\Gamma^\mathrm{ex}_0
=\Gamma$. Finally, we also assume that the extended curve
$\Gamma_d^\mathrm{ex}$ admits the  global Frenet frame, i.e. the
triple of vectors $(t(s),b(s),n(s))$ for any $s \in I_d$.
\begin{remark} \label{piecewise}
The tangetial, binormal and normal vectors are, by assumption,
$C^{2}$ functions mapping from $I_d$ to $\R^3$. The assumption
about global existence of the (unique) Frenet frame is satisfied
provided $\frac{\D}{\D s}\Gamma^\mathrm{ex}_d(s) \neq 0$ for all
$s\in I_d$. Let us emphasize, however, that we adopt this
hypothesis for simplicity only. Our main result requires only
piecewise  existence of the Frenet frame from which a global
coordinate transformation we need can be constructed rotating the
coordinate frame on a fixed angle if necessary. A discussion how
this can be done curves with straight segment can be found in
\cite{EK4}.
\end{remark}

\noindent The extended curve $\Gamma^\mathrm{ex}_d$ is, uniquely
up to Euclidean transformations, determined by its curvature
$\gamma_d^\mathrm{ex}$ and torsion $\tau_d^\mathrm{ex}$. The same
quantities for the original curve $\Gamma$ are respectively
denoted as $\gamma $ and $\tau$

\medskip

\noindent \emph{`Thin' neighborhoods of $\Gamma^\mathrm{ex}_d$.}
Consider a disc of radius $d$ parametrized by polar coordinates,
$\mathcal{B}_d: = \{r\in [0\,,d )\,, \varphi \in [0\,, 2\pi )\}$.
Using it, we define the cylindrical set $\mathcal{D}^\mathrm{ex}_d
:= I_d \times \mathcal{B}_d$ and the map $\phi_d\,:\,
\mathcal{D}_d^\mathrm{ex} \rightarrow \R^3$
$$
\phi_d (s\,, r\,,\varphi )= \Gamma^\mathrm{ex}_d (s) - r[n(s)\cos
(\varphi -\beta (s)) +b(s)\sin (\varphi -\beta (s))]\,,
$$
where the function $\beta$ will be specified latter. For $d$ small
enough the function $\phi_d$ is injective and its image
determines a tubular neighborhood $\Omega_d $ of
$\Gamma_d^\mathrm{ex}$.

The geometry of $\Omega_d$ can be described in terms of the metric
tensor written in the matrix form as
$$
(g_{ij})= \left( \begin{array}{ccc}
  h^2+r^2\zeta^2 & 0 & r^2\zeta  \\
  0 & 1 & 0 \\
  r^2\zeta  & 0 & r^2
\end{array} \right)\,,
$$
where $\zeta := \tau -\beta_{,s}$ and $h:=1+r\gamma \cos (\varphi
-\beta)$; we employ the shorthand $\beta_{,s}$ for the derivative
of $\beta$ with respect to the variable $s$. Choosing
$\beta_{,s}=\tau$ we can achieve that the metric tensor becomes
diagonal, $g_{ij}= \mathrm{diag} (h^2,1,r^2)$. This means we
choose what is usually called a \emph{Tang frame}, a coordinate
system which rotates with respect to the Frenet triple with the
angular velocity equal to the curve torsion.

The volume element of $\Omega_d$ can be expressed in the
coordinates $q\equiv (q_1,q_2,q_3)=(s,r,\varphi)$  as
$\mathrm{d}\Omega_d = g^{1/2}\mathrm{d}q$ where $g:=  \det
g_{ij}$. The following elementary inequality will be useful in the
further discussion,
\begin{equation}\label{eq-ineqh}
  |h-1|\leq d \max \gamma\,.
\end{equation}

\medskip

\noindent  \emph{Shifted curves.} Keeping  in mind a latter
purpose we define now a family of `shifted' curves
$\tilde{\Gamma}_d^\mathrm{ex} (\rho)$ located in the distance
$\rho\in(0,d]$ from $\Gamma^\mathrm{ex}_d$. Using the Frenet frame
we define $\tilde{\Gamma}^\mathrm{ex}_d(\rho)$ as graph of the
function
$$
\tilde{\Gamma}_d^\mathrm{ex} +\eta_n n+\eta_b b\,:\, (-d\,, L+d)
\to \R^3 \,,\quad \sqrt{|\eta_n|^2 + |\eta_b|^2}= \rho\,.
$$
Following the above introduced convention we use the symbol
$\tilde{\Gamma}(\rho)= \tilde{\Gamma}^\mathrm{ex}_0(\rho)$ for the
curves shifted with respect to the original $\Gamma$. Although we
do not mark it explicitly, one has to keep in mind that a shifted
curve depends  not only on the distance $\rho$ but also on the
angular variable encoded in the parameters $\eta_n,\,\eta_b$.

\medskip

\noindent Let us also list some \emph{notation} we are going to
use:

\medskip

\noindent $\bullet$ Let $\mathcal{A}\subset \R^3$ be an open set.
We use the abbreviation $( \cdot , \cdot)_{\mathcal{A}}$ for the
scalar product $( \cdot , \cdot)_{L^2(\mathcal{A}, \mathrm{d}x)}$.
If $\mathcal{A}=\R^3$ we shortly write $( \cdot , \cdot) = ( \cdot
, \cdot)_{L^2(\R^3, \mathrm{d}x)}$.

\smallskip

\noindent $\bullet$ We denote by $\mathcal{D}^\mathrm{ex}_d= I_d
\times \mathcal{B}_d$ the tubular neighborhood of the extended
curve $\Gamma^\mathrm{ex}_d$, and similarly, $\mathcal{D}_d=
I\times \mathcal{B}_d$ corresponds to the original curve $\Gamma
$.

\smallskip

\noindent $\bullet$ Given a self-adjoint operator $A$, we denote
by $\lambda_j (A)$ its $j$th eigenvalue.

\subsection{Definition of Hamiltonian and the Birman-Schwinger principle}

\emph{Boundary conditions.} The definition of the singular
Schr\"odinger operator presented below is a summary of the
discussion provided in~\cite{EK4} which we include the make this
article self-contained; we refer to the mentioned paper for more
details. Suppose given a function $f\in W^{2,2}_{\mathrm{loc}}
(\R^3\setminus\Gamma)$, its restriction to $\tilde{\Gamma }(\rho
)$ is well defined as a distribution from $D'(0,L)$ which we
denote as $f\!\upharpoonright_{\tilde{\Gamma }(\rho)}$.
Furthermore, we assume that the following limits
\begin{eqnarray}\label{eq-bc1}
 &&  \Xi(f)(s): =- \lim_{\rho \to 0 } \frac{1}{\ln \rho }
   f \upharpoonright _{\tilde{\Gamma }(\rho
  )}(s)\,,
\\ \label{eq-bc2} && \Omega (f) (s):=      \lim_{\rho \to 0 } \left[
f \upharpoonright _{\tilde{\Gamma } (\rho
  )}(s) + \Xi(f)(s)\ln  \rho \right]\,
\end{eqnarray}
exist a.e. in $(0\,,L)$. We write $f\in
\mathrm{bc}(\alpha,\Gamma)$ if a $f\in W^{2,2}_{\mathrm{loc}}(\R^3
\setminus \Gamma  )$ satisfies
\begin{equation}\label{eq-bc}
2\pi \alpha \Xi (f)=\Omega (f)\,.
\end{equation}
Equation (\ref{eq-bc}) plays the role of generalized boundary
conditions, \cite{BG85}. Then we define the set
$$
D(H_{\alpha,\Gamma}):= \{f\in W^{2,2}_{\mathrm{loc}} (\R^3
\setminus \Gamma )\cap L^2 \,:\, f\in
\mathrm{bc}(\alpha,\Gamma)\}\,
$$
and  the operator $H_{\alpha,\Gamma}\,:\, D(H_{\alpha ,\Gamma
})\to L^2$ which acts as
$$
H_{\alpha,\Gamma}f (x)=-\Delta f(x)\,,\quad x\in \R^3 \setminus
\Gamma \,.
$$
  This operator is self-adjoint,~cf.~\cite[Thm.~2.3]{EK4} and
defines the Hamiltonian  we are going to study.

\medskip

\noindent \emph{Free resolvent kernel.} We start with the
resolvent of the `free' Laplacian, $-\Delta \,:\, W^{2,2}(\R^3)\to
L^2$. It is well known that $R(-\kappa^2)=(-\Delta +\kappa^2
)^{-1}$ is for any $\kappa >0$ an integral operator with the
kernel
\begin{equation}\label{eq-kernel}
  G(\kappa; x,y) = \frac{1}{4\pi }\frac{\mathrm{e}^{-\kappa |x-y|}}{|x-y|}\,.
\end{equation}
In the following we also use the notation $G(\kappa; \rho)=
\frac{1}{4\pi}\frac{\mathrm{e}^{-\kappa \rho}}{\rho}$ where $\rho
>0$. It is  well known, see for example~\cite{BEKS}, that the operator $R(-\kappa^2)$ admits the embedding into $L^2 (I )$.
To be more precise, consider an $\omega\in L^2 (I)$ and define
\begin{equation}\label{eq-convol}
f= f_\kappa ^\omega = G(\kappa ) \omega \ast  \delta_\Gamma
=\frac{1}{4\pi } \int_{I} \frac{\mathrm{e}^{-\kappa |\cdot \,
-\Gamma (s)|}} {|\cdot \, -\Gamma (s)|}\,\omega (s)\mathrm{d}s\,.
\end{equation}
Then $f\in W^{2,2}_{\mathrm{loc}} (\R^3 \setminus \Gamma ) \cap
L^2$ and  the limit $\Omega(f)$ defines one-parameter family of
operators $\R_+ \ni \kappa \mapsto Q_{-\kappa^2}\,:\, L^2 (I) \to
L^2 (I)$ acting as
\begin{equation}\label{eq-defQ}
Q_{-\kappa^2} \omega = \Omega (f^\omega )\,,\quad \omega \in L^2
(I)\,,
\end{equation}
cf.~\cite{EK2} for more details.

\medskip

\noindent \emph{The Birman-Schwinger principle.} The stability of
the essential spectrum,
$$
\sigma _{\mathrm{ess}} (H_{\alpha,\Gamma}) = \sigma
_{\mathrm{ess}}  (-\Delta )= [0\,,\infty)\,,
$$
is a consequence of the fact that the singular potential in our
model is supported by a compact set. Using the results of
\cite{Po4} we can formulate conditions for the existence of
discrete eigenvalues. Specifically, we have,
$$
\lambda = -\kappa ^2  \in \sigma _{\mathrm{d}} (H_{\alpha ,\Gamma
}) \,\,\Leftrightarrow \,\,\ker (Q_{-\kappa^2} -\alpha)\ne
\emptyset
$$
and the multiplicity of $\lambda $ is equal to $\dim
\ker(Q_{-\kappa^2} -\alpha)$. Moreover, the corresponding
eigenspaces are spanned by the functions
\begin{equation}\label{eq-efun}
f= G(\kappa ) \omega  \ast \delta_\Gamma\,,\quad \omega \in \ker
(Q_{-\kappa^2} -\alpha )\,.
\end{equation}

\section{Main result and the proof scheme} \label{s: main}

Now we are in position to state the main result of this paper.
\begin{theorem} \label{th-main}
Let $H_{\alpha,\Gamma}$ be the singular Schr\"odinger operator
defined by means of the boundary conditions \eqref{eq-bc}
corresponding to a finite, non-closed $C^4$ smooth curve with
regular ends which has the global Frenet frame. \\  (i) The
cardinality of the discrete spectrum admits the same asymptotics
as in the case of the closed curved,~i.e.
\begin{equation}\label{eq-asmpN}
\sharp \sigma_{\mathrm{d}} (H_{\alpha ,\Gamma
})=\frac{L}{\pi}\,\zeta_\alpha (1+\mathcal{O}(\mathrm{e}^{\alpha
\pi }))\,,
\end{equation}
where
$$
\zeta_\alpha = (-\xi_\alpha )^{1/2}\,.
$$
(ii) Furthermore, the $j$th eigenvalue of $H_{\alpha,\Gamma}$ has
the expansion
\begin{equation}\label{eq-mainasym}
\lambda_j (H_{\alpha,\Gamma})= \xi_\alpha +\lambda _j
(S)+\mathcal{O}(\mathrm{e}^{\pi \alpha })\quad \mathrm{for}\quad
\alpha \to -\infty\,,
\end{equation}
where $\lambda _j(S)$ stands for the $j$th eigenvalue of the
operator
\begin{equation}\label{eq-defS}
S=-\frac{d^2}{ds^2 }- \frac14 \gamma^2(s)\,:\, D(S)\rightarrow L^2
(0, L)
\end{equation}
with the domain $D(S):= W^{1,2}_0 (0,L)\cap W^{2,2} (0,L)$.
\end{theorem} 
\begin{remark}
\rm{One may also ask a question on a varying interaction. Of
course, there is not a unique answer for a general case, however,
admitting a varying coupling $\tilde{\alpha} = \alpha +\omega (s)
$, where $\omega (s) \in C^2_0 (\R)$  instead of $\alpha $, we may
expect the asymtotics
\begin{equation}
\lambda_j (H_{\tilde{\alpha},\Gamma})= \lambda _j
(\tilde{S})+\mathcal{O}(\mathrm{e}^{\pi \alpha })\quad
\mathrm{for}\quad \alpha \to -\infty\,,
\end{equation}
} where $$ \tilde{S}= -\frac{d^2 }{ds^2}-\frac{\gamma ^2 (s)
}{4}-4\e^{2(-2\pi \tilde{\alpha}(s)+\psi (1) )}\,.
$$
However, the model requires detailed analysis and further
generalizations of the methods used in this paper.
\end{remark}
\subsection{The proof scheme}

\emph{Dirichlet--Neumann bracketing.} The asymptotics
(\ref{eq-mainasym}) can not be obtained directly from the
Dirichlet--Neumann bracketing on tubular neighborhoods of the
curve $\Gamma$ in the way analogous to the loop case \cite{Ex08,
EK4}, because in the lower bound the operator $S$ would be
replaced by the operator $S^N$ acting as (\ref{eq-defS}) but
Neumann boundary conditions. Nevertheless, this technique is
powerful enough to yield claim (i) of the theorem.  More
specifically, using the Dirichlet--Neumann bracketing and
repeating the argument of ~\cite{EK4} we get
\begin{equation}\label{eq-asyp}
 \lambda_j  (H_{\alpha,\Gamma})= -  \kappa_j (\alpha )^2=\xi_\alpha +c_j +
  \mathcal{O}(\mathrm{e}^{\pi \alpha })\,,
\end{equation}
where the numbers $c_j$ satisfy the inequalities
\begin{equation} \label{eq-evestim}
\lambda_j (S^N)\leq c_j \leq \lambda_j (S)\,,
\end{equation}
and $S^N\,:\, D(S^N)=\{f\in W^{2,2}(0,L)\,:\, f'(0)=f'(L)=0 \}\to
L^2(0,L)$; recall that $S^N$ has the same differential symbol as
$S$. Note that the second inequality of (\ref{eq-evestim})
reproduces a right upper bound. In order to prove
Theorem~\ref{th-main} we obviously have to replace the first
inequality by a better lower bound. The remaining part of the
paper is devoted to this problem.

\medskip

\noindent \emph{A few ideas.} Let us mention three concepts we are
going to use in the proof of Theorem~\ref{th-main}. The first is
the observation that the properties of the discrete spectrum are
reflected in the behavior of the eigenfunctions in the vicinity of
the curve $\Gamma$. Specifically, let $f_j$ stand for the $j$th
eigenfunction of $H_{\alpha,\Gamma}$ corresponding to $\lambda_j
(H_{\alpha,\Gamma})$. Then we have
$$ 
\lambda_j (H_{\alpha,\Gamma})=\frac{(H_{\alpha,\Gamma}f_j,f_j )}
{\|f_j \|^2} =\frac{(-\Delta_{\alpha ,\Gamma }f_j,f_j
)_{\Omega_d}} {\|f_j \|_{\Omega_d}^2} \,,
$$ 
where the second one of the equalities follows from the natural
embedding $L^2(\R^3) \supset L^2 (\Omega_d)$ in combination with
the fact that $f_j$ satisfies away from $\Gamma$ the appropriate
differential equation: the symbol $-\Delta_{\alpha,\Gamma }$ is
understood not as a self-adjoint operator, rather as the
differential expression, $-(\Delta_{\alpha,\Gamma}f)(x)=-(\Delta
f)(x)$ for $x\neq \Gamma $ and $f\in W^{2,2}_\mathrm{loc}
(\Omega_d)$.

The second idea  is to employ a suitable `straightening'
transformation which allows us to translate the geometry of the
problem into the coefficients of the operator. In particular, we
obtain an effective potential expressed in terms of the curvature
of $\Gamma$ and its derivatives. To this aim we introduce two
unitary transformations,
$$
Uf = f\circ \phi_d \,:\, L^2 (\Omega_d) \rightarrow L^2
(\mathcal{D}^\mathrm{ex}_d ,g^{1/2} \mathrm{d}q )
$$
and the other one removing the weight factor in the inner product,
$$
\hat{U}f =  g^{1/4}f\,,\quad  \hat{U}\,:\, L^2
(\mathcal{D}^\mathrm{ex}_d ,g^{1/2} \mathrm{d}q ) \rightarrow L^2
( \mathcal{D}_d^\mathrm{ex}, \mathrm{d}q )\,;
$$
we combine them denoting
\begin{equation}\label{eq-deffg}
f^g :=\hat{U} Uf\,.
\end{equation}
Since $f_j$ is by assumption the $j$th eigenfunction of
$H_{\alpha,\Gamma}$, in view of (\ref{eq-bc}) we have
$g^{-1/4}f_j^g \in\mathrm{bc} (\alpha,\Gamma)$. After a
straightforward calculation~\cite{EK4}, we get
\begin{equation} \label{eq-estimHam}
(-\Delta_{\alpha,\Gamma }f_j ,f_j )_{\Omega_d}= \left(
(-\partial_s h^{-2}\partial_s +T_\alpha +V )f^g_j , f^g_j
\right)_{\mathcal{D}_d^\mathrm{ex}}
\,,
\end{equation}
where $T_\alpha $ is defined by the differential expression
\begin{equation}\label{eq-defT}
T_\alpha =-\partial_r^2 - r^{-2}\partial^2_\varphi
-\frac{1}{4}r^{-2}\,
\end{equation}
and
\begin{equation}\label{eq-defV}
V=-\frac{\gamma
^2}{4h}+\frac{h_{,ss}}{2h^3}-\frac{5(h_{,s})^2}{4h^4}\,.
\end{equation}
 Note that the above described idea was used, for example,
in the context of waveguides, cf.~\cite{DE, KA}. 

 Finally, the third concept is to use is an
approximation $f_j\upharpoonright_{\Omega_d }$ by functions
vanishing on $\partial \Omega_d$. To explain why it is possible
note that in view of $f_j \in \mathrm{bc}(\alpha,\Gamma)$ the
eigenfunctions have a logarithmic singularity at $\Gamma $,
however, away from the curve they decay rapidly:
relations~(\ref{eq-kernel}) and (\ref{eq-efun}) show that $f_j (x)
\sim \mathrm{e}^{-\kappa_j (\alpha )|x-\Gamma|}$ holds for $x\in
\Omega_d \setminus \Gamma $, where $\kappa _j (\alpha
):=\sqrt{-\lambda_j(\alpha )}$. It shows, in particular, that
$f_j$ `accumulates' at the curve $\Gamma $ as $\alpha \to-\infty$.
This suggests that one might get a good estimate replacing
$f_j\upharpoonright_{ \Omega_d }$ by   suitable functions
vanishing on $\partial \Omega_d$ and relate simultaneously the
transverse size of $\Omega_d$ to the parameter $\alpha $. To this
aim we assume in the following that
\begin{equation}\label{eq-defd}
d=d (\alpha )=\mathrm{e}^{\pi \alpha }\,.
\end{equation}

\medskip

\noindent \emph{Proof steps.} We are going to use the described ideas in the following way: \\[.2em]
$\bullet$ We construct a self-adjoint operator $W$ in
$L^2(\mathcal{D}_d^\mathrm{ex})$ acting as
\begin{equation}\label{eq-Wdef}
W=-\partial_s h^{-2}\partial_s+T_\alpha +V\,:\, D(W) \rightarrow
L^2 (\mathcal{D}^\mathrm{ex}_d)
\end{equation}
with the domain $D(W)$ consisting of the functions that satisfies
Dirichlet boundary conditions on $\partial
\mathcal{D}^{\mathrm{ex}}_d$ and $g^{-1/4} f\in
\mathrm{bc}(\alpha,\Gamma)$. Our aim is to find a lower bound for
eigenvalues of $H_{\alpha,\Gamma }$ in terms of $W$. Specifically,
we are going to show that the following asymptotic inequality
\begin{equation}\label{eq-evWmain}
\lambda_j (W)\leq \lambda_j (H_{\alpha,\Gamma})
+\mathcal{O}(d^{-18}\mathrm{e}^{-C/d})\,
\end{equation}
holds.

\smallskip

\noindent $\bullet$ The next step is to recover a lower bounds for
$\lambda_j (W) $. Using a variational argument we prove that
\begin{equation}  \label{eq-lbW}
\lambda_j (W ) \geq  \xi_\alpha +\lambda_j (S)+\mathcal{O}(d)\,.
\end{equation}
Combining it with (\ref{eq-evWmain}) and (\ref{eq-evestim}) we
obtain the claim of~Theorem~\ref{th-main}.

\section{Approximating $f_j$ by Dirichlet functions}
\label{s: Dirichlet approx}

For the sake of brevity we shall speak of the functions $f\in
D(W)$ involved in the first step as of Dirichlet functions; we are
sure that the reader would not confuse them with other objects
bearing in mathematics the same name.

We keep the notation $\lambda_j (H_{\alpha,\Gamma})=
-\kappa_j(\alpha)^2$ for the eigenvalues of our original operator.
To investigate the behavior of the corresponding eigenfunction
$f_j$ we employ the expression
$$
f_j=G(\kappa_j (\alpha ))\omega_j \ast \delta_\Gamma \,,
$$
where
\begin{equation}\label{eq-BS3}
(Q_{-\kappa _j(\alpha )^2}-\alpha )\omega_j =0\,,
\end{equation}
following from~(\ref{eq-efun}). For brevity again we shall write
in the following shortly $f_j=G(\kappa_j(\alpha))\omega_j $;
without loss of generality we may assume that $\omega _j$ is
normalized function, i.e.~$\|\omega_j \|_{I}=1$. Combining
(\ref{eq-defd}) and (\ref{eq-asyp}) we get
\begin{equation}\label{eq-asymsought}
  \kappa_j (\alpha )d (\alpha )= Cd(\alpha )^{-1} + \mathcal{O}(d(\alpha )^{3})
\end{equation}
with the constant $C:=2\mathrm{e}^{2\psi(1)} $.

\subsection{Approximate orthogonality of Dirichlet functions}

Now we approximate the eigenfunctions $f_j$ by suitable Dirichlet
functions. We set $f_j^D=\eta f_j^g $, where $\eta \in C^{\infty
}_0 (\mathcal{D}^\mathrm{ex}_d)$ is a positive function such that
$$
\eta (x)=1 \quad \mathrm{for} \;\:  x\in
\mathcal{D}^\mathrm{ex}_{d/2}
$$
and $f^g$ is the `straightened' function defined
by~(\ref{eq-deffg}).

\begin{lemma} \label{le-ort}
Let $d$ be given by (\ref{eq-defd}), then the following
asymptotics,
\begin{equation}\label{eq-ort}
  (f^D_j, f^D_k)_{ \mathcal{D}^\mathrm{ex}_d}=
  \|f_j^D \|^2_{\mathcal{D}^\mathrm{ex}_d}\delta_{jk} + R(d)\quad
  \mathit{as} \;\; \alpha \to -\infty\,,
\end{equation}
holds with the remainder term satisfying
$$
|R(d)|= \mathcal{O}(d^{-2}\mathrm{e}^{-C/d})\,.
$$
\end{lemma}
\begin{proof}
We start from the self-evident statement that
$$
(f_j, f_k)= \|f_j \|^2  \delta_{jk}= (f_j , f_k)_{\Omega _d}
+(f_j, f_k)_{\R^3 \setminus \Omega_d}\,.
$$
In view of the unitarity of the straightening transformation the
first term on the right-hand side can be written as $(f_j ,
f_k)_{\Omega_d}= (f^g_j, f^g_k)_{\mathcal{D}^\mathrm{ex}_d}$ which
implies
\begin{equation}\label{eq-aux1}
(f^g_j , f^g_k)_{\mathcal{D}^\mathrm{ex}_d}= \|f_j\|^2 \delta_{jk}
-(f_j , f_k)_{\R^3 \setminus \Omega_d}\,.
\end{equation}
The remaining part of the argument can be divided into two parts:

\smallskip

\noindent \emph{Step 1. Approximating $f^g_k $ by means of
$f^D_k$.} Consider a point$x\in\Omega_d $ and denote
$x_q:=\phi_d(x)$. Combining the inequality $|x_q-\Gamma (s)| \geq
r$ with \eqref{eq-convol} and  (\ref{eq-ineqh}) we obtain
\begin{eqnarray} \nonumber
\lefteqn{\left| ( f^g_j ,f^g_k)_{\mathcal{D}^\mathrm{ex}_d} -
( f^D_j ,f^D_k)_{\mathcal{D}^\mathrm{ex}_d} \right|} \\
\nonumber  && = \int_{\mathcal{D}^\mathrm{ex}_d \setminus
\mathcal{D}^\mathrm{ex}_{d/2}} g^{1/2}(1-\eta^2)
G(\kappa_{j}(\alpha ))
\omega_j \, \overline{ G(\kappa_{k}(\alpha ))\omega_k} \,\mathrm{d}q \\
\nonumber && \le C_1 \int_{d/2}^{d} \int_{0}^{2\pi}\int_{I_d}
G(\kappa_{j}(\alpha )) \omega_j \, \overline{ G(\kappa_{k}(\alpha
)) \omega_k}
 \,r\,\mathrm{d}r \,\mathrm{d}\varphi \,\mathrm{d} s \\
\nonumber  && \le \frac{C_1}{2}\, (2d+L)\, \|\omega_j \|_{L^1 (I)}
\|\omega_k \|_{L^1 (I)} \int_{d/2}^d \frac{\mathrm{e}^{-(\kappa_j
(\alpha )+\kappa_k(\alpha ))r }}{r}\,
   \mathrm{d}r \\  \label{eq-estim3} && \le \frac{C_1}{2}\, L^2(2d+L)\,
\mathrm{e}^{-C/d}
\end{eqnarray}
with some constant $C_1>0$. The last estimate comes from
(\ref{eq-asymsought}) and Schwartz inequality which gives
$\|\omega_j \|_{L^1 (I)}\leq  L\|\omega_j\|_I=L$ for any
$j\in1,\dots,N$; we have also used here $|I_d|=2d+L$. Combining
(\ref{eq-aux1}) and (\ref{eq-estim3}) we get
\begin{equation}\label{eq-1and2version2}
  (f_j^D, f_k^D)_{\mathcal{D}^\mathrm{ex}_d}=
  \|f_j^D\|^2_{\mathcal{D}^\mathrm{ex}_d} \delta_{jk}
  +\|f_j\|^2_{\R^3 \setminus \Omega_d }\delta_{jk} -(f_j,f_k)_{\R^3\setminus\Omega_d}
   + R_{1}(d)
\end{equation}
with the remainder term satisfying
$$
|R_{1}(d)| = \mathcal{O}(\mathrm{e}^{-C/d})\,;
$$
it remains to estimate the parts of (\ref{eq-1and2version2})
referring to $L^2 (\R^3  \setminus \Omega_d )$.

\smallskip

\noindent\emph{Step 2. Estimates of $\|f_j\|_{\R^3 \setminus
\Omega_d}$.} Consider the ball $\mathcal{B}=B(\Gamma (L/2)\,, L)$
of the radius $L$ centered at $\Gamma (L/2)$, the midpoint of the
curve. For $d$ small enough we obviously have $\Omega_d \subset
\mathcal{B}$, and consequently, we can decompose the norm
$\|f_j\|_{\R^3 \setminus \Omega_d}$ as
\begin{equation}\label{eq-decomnorm}
\|f_j\|^2_{\R^3 \setminus \Omega_d}=\|f_j\|^2_{\R^3 \setminus
\mathcal{B}}+ \|f_j\|^2_{ \mathcal{B}\setminus \Omega_d}\,.
\end{equation}
Let us  introduce the spherical coordinates $(\hat{r},
\hat{\theta}, \hat{\varphi})$, where $\hat{r}$ is the radius
measuring the distance from the ball center at $\Gamma (L/2)$ and
$\hat{\theta},\hat{\varphi}$ are appropriate polar and azimuthal
angles. Employing the inequality $|x-\Gamma(s)| \geq \hat{r}-L/2$
for $x\in \R^3 \setminus \mathcal{B} $ we get by a straightforward
computation
\begin{eqnarray}\nonumber
 \lefteqn{\|f_j\|^2_{\R^3 \setminus
\mathcal{B}} = \int_{\R^3 \setminus \mathcal{B}} \left| \int_{I}
\,\frac{\mathrm{e}^{ -\kappa_j |x-\Gamma (s)| }} {4\pi |x-\Gamma
(s)|}\, \omega_j (s)\mathrm{d}s\right|^2\mathrm{d}x} \\ && \leq
\nonumber \|\omega_j\|_{L^1 (I)}^2\int_{0}^{2\pi} \int_{0}^\pi
\int_{L}^\infty \left( \frac{\hat{r}}{4\pi( \hat{r}  - L/2 )}
\right)^{2} \mathrm{e}^{-2\kappa_j (  \hat{r}- L/2) }
\mathrm{d}\hat{r} \mathrm{d}\hat{\theta }\mathrm{d}\hat{\varphi}
\\ && \leq \label{eq-auxB} \frac {L^2}{16\kappa_j }
\mathrm{e}^{-L\kappa_j}= \mathcal{O}(d^2\,
\mathrm{e}^{-CL/d^2})\,,
\end{eqnarray}
where we have used (\ref{eq-asymsought}) and $\|\omega_j
\|_{L^1(I)}\leq L$. The second norm at the right-hand side of the
decomposition (\ref{eq-decomnorm}) can estimated as
\begin{eqnarray}\nonumber
\lefteqn{\|f_j\|^2_{\mathcal{B}\setminus \Omega_d} =
\int_{\mathcal{B}\setminus \Omega_d} \left| \int_{I}\,
\frac{\mathrm{e}^{ -\kappa_j |x-\Gamma (s)| }} {4\pi |x-\Gamma
(s)|}\,\omega_j (s)\mathrm{d}s\right|^2\mathrm{d}x} \\ && \leq
\label{eq-auxA} \mathrm{vol}(\mathcal{B}\setminus \Omega_d)\,
\frac{\mathrm{e}^{-2\kappa_j d}}{(4\pi d)^2}\, \|\omega_j
\|^2_{L^1 (I)} = \mathcal{O}(d^{-2}\mathrm{e}^{-C/d})\,.
\end{eqnarray}
Combining (\ref{eq-auxB}) and (\ref{eq-auxA}) we get
$$
\|f_j \|^2_{\R^3 \setminus \Omega_d} = \mathcal{O} (d^{-2}
\mathrm{e}^{-C/d})\,,
$$
which together with the result of the first step yields the sought
claim.
\end{proof}

\subsection{Estimates for the operator $W$} \label{s: est W}

We also have to find how the `Dirichlet trimming' influences the
operator $W$ defined by \eqref{eq-Wdef}. The idea of replacing the
true `straightened' eigenfunctions $f_j^g$ by the Dirichlet
approximants is based on the fact that the contribution coming
from
$$
\breve{\mathcal{D}}_d:= \mathcal{D}^\mathrm{ex}_d\setminus
\mathcal{D}^\mathrm{ex}_{d/2}
$$
is asymptotically negligible. Note that, on the one hand, the
operator $W$ acts up to the unitary transformation $\hat{U}U$ as
$H_{\alpha,\Gamma}$ on the functions supported by
$\mathcal{D}_{d/2}^\mathrm{ex}$. On the other hand, the following
two lemmata justify the just made claim by gauging the component
coming from $\breve{\mathcal{D}}_d$.

\begin{lemma} \label{le-remain}
The asymptotical relation
\begin{equation}\label{eq-asmpW}
|(Wf^D_j ,f^D_k )|_{\breve{\mathcal{D}}_d} =
\mathcal{O}(d^{-8}\mathrm{e}^{-C/d})
\end{equation}
holds for $d$ defined by (\ref{eq-defd}) and $\alpha \to -\infty$.
\end{lemma}
\begin{proof}
We start from an elementary Schwarz inequality estimate,
$$
|(Wf^D_j,f^D_k )_{\breve{\mathcal{D}}_d} | \leq
\|Wf^D_j\|_{\breve{\mathcal{D}}_d} \|f^D_k
\|_{\breve{\mathcal{D}}_d}\,.
$$
Proceeding in analogy with Step 1 in the proof of
Lemma~\ref{le-ort}, cf.~(\ref{eq-estim3}), we get for the norm
$\|f_k ^D\|^2_{\breve{\mathcal{D}}_d}$ the bound
$$
\|f_k ^D\|^2_{\breve{\mathcal{D}}_d} = \|\eta f_k ^g\|^2
_{\breve{\mathcal{D}}_d} =  \int _{\breve{\mathcal{D}}_d}
g^{1/2}|\eta G(\kappa_j (\alpha ) ) \omega_j |^2 \mathrm{d}q =
\mathcal{O}( \mathrm{e}^{-C/d})\,.
$$
Next we  estimate $\|Wf^D_k\|_{\breve{\mathcal{D}}_d}$. Applying
 (\ref{eq-estimHam}) we obtain
\begin{eqnarray} \label{eq-W}
\lefteqn{Wf^D_k = W(\eta f^g _k )=  \eta \left( -\partial_s
h^{-2}\partial_s +T_\alpha +V  \right) f^g _k}  \\
\label{eq-estimB} && +
 \left( -\partial_{q_1} \partial_{\tilde{q}_1}\eta - \sum_{i=2}^3
\frac{d^2}{d\tilde{q}_i^2} \eta \right) f^g _k
 \\ && \label{eq-estimC}  - (\partial_{q_1}\eta )
 (\partial_{\tilde{q}_1}f^g_k )  - (\partial_{q_1}f_k^g )
 (\partial_{\tilde{q}_1} \eta  )
 \\ && \label{eq-estimD}  -2\sum_{i=2}^3\partial_{\tilde{q}_i} \eta
\partial_{\tilde{q}_i} f^g_k \,,
\end{eqnarray}
where we use the shorthands $\partial_{\tilde{q}_1}
=h^{-2}\partial_s$, $\partial_{\tilde{q}_2} =\partial_r$, and
$\partial_{\tilde{q}_3}=\frac{1}{r} \partial_\varphi$, and the
involved differential expressions have been defined in
(\ref{eq-defT}) and (\ref{eq-defV}). Since $(-\partial_s
h^{-2}\partial_s +T_\alpha +V) f^g_k (q)= \lambda_k
(H_{\alpha,\Gamma})f^g_k (q)$ holds for $q\in
\mathcal{D}^\mathrm{ex}_d$ and $|\eta| $ is bounded by assumption,
the norm of the right-hand-side expression of (\ref{eq-W}) can be
estimated by means of $|\lambda_k (H_{\alpha,\Gamma })|^2
\|f_k^g\|^2_{\mathcal{D}_d^\mathrm{ex}}$. Moreover, it is easy to
see that the factor appearing in the longitudinal part of the
operator satisfies $h^{-2}=1+\mathcal{O}(d)$ and $\partial_s
h^{-2}=\mathcal{O}(d)$ as $d\to 0$, which implies
$|\partial_{q_1}\partial_{\tilde{q}_1} \eta |\leq
\mathrm{const}\,d^{-2}$. Using further inequality
$|\partial^2_{\tilde{q}_i} \eta |\leq \mathrm{const}\,d^{-2}$,
$i=2,3$, we can estimate the norm of the expression
(\ref{eq-estimB}) by means of $d^{-4}
\|f_k^g\|^2_{\mathcal{D}_d^\mathrm{ex}}$.

Suppose that $x_q\in \Omega_d$. We put again $x_q = \phi_d (x)$
and denote $\rho (q;s'):=|x_q-\Gamma (s')|$. Since
$|\partial_{\tilde{q}_i}\rho| \leq \mathrm{const}\, d^{-1} $ we
have
$$
\left|\partial_{\tilde{q}_i} \frac{\mathrm{e^{-\kappa
\rho}}}{\rho}\right| \leq \mathrm{const}\,\frac{\mathrm{e^{-\kappa
\rho }}}{\rho^2}\left( \kappa +\frac{1}{\rho } \right)\,.
$$
Applying the above inequality to the expression (\ref{eq-estimC})
and combining this with the fact that the quantity
$|\partial_{\tilde{q}_i}\eta|$ entering (\ref{eq-estimD}) is
bounded by $\mathrm{const}\,d^{-1}$ we get
\begin{eqnarray} \nonumber
\lefteqn{\|W(\eta f^g _k )\|^2_{\breve{\mathcal{D}}_{d}}  \leq C_3
|\lambda_k (H_{\alpha ,\Gamma
})|^2\|f_k^g\|^2_{\mathcal{D}_d^\mathrm{ex}} +C_3'\,d^{-4}
\|f_k^g\|^2_{\mathcal{D}_d^\mathrm{ex}}} \\ && + C_3
''d^{-2}\|\omega_k \|^2_{L^1 (I)} \int_{d/2}^d
\frac{\mathrm{e^{-2\kappa r }}}{r^4}\left( \kappa +\frac{1}{r }
\right)^2 r\,\mathrm{d}r = \mathcal{O}(d^{-8} \mathrm{e}^{-C/d})
\phantom{AAA}
\end{eqnarray}
with appropriate constants. It completes the proof.
\end{proof}

The aim of the next lemma is to find out a lower bound for
$\|f^D_j \|_{\mathcal{D}^\mathrm{ex}_d}$ which will give us a
possibility to compare this norm with the small terms appearing in
relations (\ref{eq-ort}) and (\ref{eq-asmpW}).

\begin{lemma} \label{le-lower}
Let $d$ be given by (\ref{eq-defd}). Then there exists a $c>0$
such that
\begin{equation}\label{eq-upper}
 \|f_j^D \|^2_{\mathcal{D}^\mathrm{ex}_d} \geq  c\, d^8\,.
\end{equation}
\end{lemma}
\begin{proof}
We inspect first the behavior of the eigenfunction $f_{j}$ in
$\Omega_d$. Combining the boundary conditions (\ref{eq-bc1}) and
(\ref{eq-bc}) with (\ref{eq-BS3}) we get
$$
f_j \!\upharpoonright_{\Gamma (r)} =G(\kappa_{j}) \omega_j
\!\upharpoonright_{\Gamma (r)} = -\frac{1}{2\pi}\, \omega_j \ln r
+ \alpha\,\omega_j + o(r)\,,
$$
where the error term on the right-hand side means a function from
$L^2 (I)$ the norm of which is $o(r)$ uniformly in $\alpha$,
cf.~\cite{EK2}. Consider the curve distances $r\in (0, d^4)$, then
in view of (\ref{eq-defd}) the inequality
$$
-\ln r > -4\pi \alpha
$$
holds for any $\alpha$, in particular, for $\alpha\to-\infty$.
This implies
\begin{equation}\label{eq-asymshifted}
  \|f_j \!\upharpoonright_{\Gamma (r)} \|^2_{I} \geq
\left( \frac{\ln r}{4\pi }\right)^2 \|\omega_j\|^2_{I}+o(\ln r)
=\left( \frac{\ln r}{4\pi }\right)^2    +o(\ln r)\,,
\end{equation}
where we have used the fact that the functions $\omega_j$ are
normalized by assumption. Consequently, for $d$ small enough we
can estimate
\begin{equation}\nonumber 
  \|f_j^D \|^2_{\mathcal{D}^\mathrm{ex}_d} \geq   \|f_j^g\|^2_{\mathcal{D}^\mathrm{ex}_{d/2}} \geq
  4\pi \int_{0}^{d^4} \|f_j \!\upharpoonright_{\Gamma (r)} \|^2_{I}\,
  r\, \mathrm{d}r \geq  c\, d^8 +o(d^8 )\,,
\end{equation}
where $c$ is a positive constant.
\end{proof}

Combining the asymptotics (\ref{eq-ort}) with the bound
(\ref{eq-upper}) we obtain
\begin{equation}\label{eq-ortA}
(f^D_j, f^D_k)_{ \mathcal{D}^\mathrm{ex}_d}= \|f_j^D
\|^2_{\mathcal{D}^\mathrm{ex}_d} \delta_{jk}
  +\|f^D_j\|_{ \mathcal{D}^\mathrm{ex}_d} \|f^D_k\|_{ \mathcal{D}^\mathrm{ex}_d}
  R_2(d)\,,
\end{equation}
where
$$
R_2 (d)=\mathcal{O}(d^{-10} \mathrm{e}^{-C/d})\,;
$$
on the other hand, a combination of (\ref{eq-asmpW}) with
(\ref{eq-upper}) yields
\begin{equation}\label{eq-asmpWA}
|(Wf^D_j ,f^D_k )|_{\breve{\mathcal{D}}_d} = \|f^D_j
\|_{\breve{\mathcal{D}}_d}  \|f^D_k \|_{\breve{\mathcal{D}}_d} R_3
(d) \,,
\end{equation}
where
$$
R_3(d) = \mathcal{O}(d^{-16} \mathrm{e}^{-C/d})\,.
$$

\section{Eigenvalues of $W$} \label{s: ev-as}

In this section we are going to conclude the proof of
Theorem~\ref{th-main} by demonstrating the inequalities
(\ref{eq-evWmain}) and (\ref{eq-lbW}).

\subsection{A lower bound for $H_{\alpha , \Gamma }$ in the terms of $W$}\label{s: squeeze}

Our first aim is to derive inequality (\ref{eq-evWmain}) in a way
partially inspired by \cite{EP13}.
\begin{lemma}
Let $d$ be given by (\ref{eq-defd}), then for $\alpha \to -
\infty$ we have
\begin{equation}\label{eqWevasym}
  \lambda_j (W)\leq \lambda_j (H_{\alpha,\Gamma})+\mathcal{O}
  (d^{-18}\mathrm{e^{-C/d}})\,.
\end{equation}
\end{lemma}
\begin{proof}
Fix a number $k\in\N$. According to the minimax principle we have
\begin{equation}\label{eq-minmax}
\lambda_k(W) = \sup_{\mathcal{S}_k} \inf_{f\in
{\mathcal{S}_k^\perp }} \frac{(Wf,f)_{\mathcal{D}^\mathrm{ex}_d}}
{\|f\|^2_{\mathcal{D}^\mathrm{ex}_d}}\,, \quad f\in D(W)\,,
\end{equation}
where  $S_k$ runs through $(k-1)$-dimensional subspaces of $L^2
(\mathcal{D}^\mathrm{ex}_d) \cap D(W)$. It follows from
Lemma~\ref{le-ort} that the functions $\{ f_j^D \}_{j=1}^N$ are
linearly independent for all $d$ small enough, and consequently,
at least one of the function from each $\mathcal{S}_k^{\perp}$
admits the decomposition
$$
h=\sum_{j=1}^k h_j f_j^D\,, \quad h_j \in \C\,,
$$
which means that
\begin{equation}\label{eq-minimaxB}
  \inf_{f\in {\mathcal{S}_k^\perp }} \frac{(W f,
f)_{\mathcal{D}^\mathrm{ex}_d}}{\|f\|^2_{\mathcal{D}^\mathrm{ex}_d}}
\leq \frac{(W h,
h)_{\mathcal{D}^\mathrm{ex}_d}}{\|h\|^2_{\mathcal{D}^\mathrm{ex}_d}}\,.
\end{equation}
Using next the fact that $(Wf_k^D)(q)
=\lambda_k(H_{\alpha,\Gamma}) f_k^D (q)$ holds for any $q\in
\mathcal{D}^\mathrm{ex}_{d/2}$ together with the asymptotic
relations (\ref{eq-asmpWA}) and (\ref{eq-ortA}) we get
\begin{equation} \label{minmaxC}
(W h,h)_{\mathcal{D}_d^\mathrm{ex}} =  \sum_{j=1}^k \lambda_j
(H_{\alpha,\Gamma})\, |h_j|^2\, \|f^D_j \|^2_{D^\mathrm{ex}_{d/2}}
+\sum_{i,j=1}^k h_i \overline{h_j }\,S_{ij}(d) \,,
\end{equation}
where
$$
S_{ij}(d):=    \|f^D_i\|_{ \mathcal{D}^\mathrm{ex}_d} \|f^D_j\|_{
\mathcal{D}^\mathrm{ex}_d}\,  \big( \lambda_j
(H_{\alpha,\Gamma})R_2(d/2)+ R_3
 (d) \big)\,.
$$
Consequently, using (\ref{eq-asymsought}), (\ref{eq-asmpW}),
(\ref{eq-asmpN}) and (\ref{eq-asyp}), we can estimate the last
term of (\ref{minmaxC}) as
\begin{eqnarray} \nonumber
\lefteqn{\Big|\sum_{i,j=1}^k  h_i \overline{h_j }S_{ij}(d) \Big|
\leq k \big(| \lambda_k (H_{\alpha,\Gamma})R_2(d/2)|+| R_3
 (d)| \big)\, \sum_{i=1}^k |h_i|^2
\|f^D_i\|^2 _{ \mathcal{D}^\mathrm{ex}_d}}  \\ && \phantom{AAAA}
\label{eq-auxest1} = R_4(d)\sum_{i=1}^k |h_i|^2 \|f_i^D\|^2 \,,
\phantom{AAAAAAAAAAAAAAAAAAAA}
\end{eqnarray}
where
$$
R_4(d) =\mathcal{O}(d^{-18}\mathrm{e}^{-C/d})\,.
$$
This yields
\begin{equation}\label{eq-estimW}
  (W h, h)_{\mathcal{D}_d^\mathrm{ex}}=  \sum_{j=1}^k \big(\lambda_j
(H_{\alpha,\Gamma})+ \mathcal{O}(d^{-18}\mathrm{e}^{-C/d}) \big)\,
|h_j|^2\|f^D_j \|^2_{D^\mathrm{ex}_{d}}\,.
\end{equation}
In the analogous way we can get an asymptotic expression for the
norm,
\begin{eqnarray} \label{eq-auxdenom}
\|h\|^2_{\mathcal{D}^\mathrm{ex}_d}= \sum_{j=1}^k\,
 \big(1+\mathcal{O}(d^{-12}\mathrm{e}^{-C/d})\big)\,
 |h_j|^2 \|f_j^D\|^2 _{\mathcal{D}^\mathrm{ex}_d} \,.
\end{eqnarray}
Combining now the relations (\ref{eq-auxest1}) and
(\ref{eq-auxdenom}), taking into account (\ref{eq-minmax}) and
(\ref{eq-minimaxB}), we arrive at the desired result.
\end{proof}

\subsection{A lower bound for $W$}

Finally, we are going to prove (\ref{eq-lbW}). It will be done in
two steps.

\noindent \emph{An auxiliary lower bound.}  Our first aim is to
show
\begin{equation}\label{eq-W-ev1}
  \lambda_j (W)\geq  \xi_\alpha +\lambda_j (S_d^\mathrm{ex}) + \mathcal{O}(d)\,,
\end{equation}
where
$$
S_d^\mathrm{ex}
=-\frac{d^2}{ds^2}-\frac14(\gamma^\mathrm{ex}_d)^2\,:
W_0^{2,2}(I_d) \to L^2 (I_d)\,.
$$
To prove this statement we recall that the operator $W$ is defined
as
$$
W=-\partial_s h^{-2}\partial_s +T_\alpha +V\,:\, D(W) \to L^2
(\mathcal{D}^{ex}_d)\,,
$$
where the functions $f\in D(W)$ satisfy Dirichlet boundary
conditions on $\partial \mathcal{D}^{ex}_d$ and $g^{-1/4}f\in
\mathrm{bc}(\alpha,\Gamma )$. In particular, those functions are
continuous away from $\Gamma $. Given an $s\in [-d, L+d]$ we
denote by $f_s \in L^2 (\mathcal{B}_d)$ the `cut' function, $f_s
(r,\varphi ) := f(s,r,\varphi )$ where  $f\in D(W) \subset  L^2
(\mathcal{D}^{ex}_d)$. Operator $T_\alpha $ can be decomposed into
a direct integral, $ T_\alpha = \int^\oplus_{[-d, L+d]} T_\alpha
(s)\,\mathrm{d}s$, on $L^2 (\mathcal{D}_d^{ex})=\int^\oplus_{[-d,
L+d]} L^2(\mathcal{B}_d) \,\mathrm{d}s$. In other words, for any
$s\in [0,L]$ the operators $T_\alpha (s)$ act as
\begin{equation} \label{eq-decompT}
T_\alpha (s)f_s = \left( - \partial^2_r  -
r^{-2}\partial^2_\varphi  - \frac{1}{4}r^{-2} \right)f_s\,,
\end{equation}
where $g_s^{1/4}f_s \in \mathrm{bc}(\alpha , r=0)$ and $f_s$
satisfies Dirichlet boundary conditions on $\partial
\mathcal{B}_d$. Furthermore, for $s\in [-d, 0)\cup (L, L+d]$
operators $T_\alpha (s)$ act as (\ref{eq-decompT}), however,
functions from their domains are regular at the origin as the
point interaction is absent at the extended parts of the curve. Of
course, they still satisfy  Dirichlet boundary conditions on
$\partial \mathcal{B}_d$. For a fixed $s \in[-d,L+d] $ we denote
by $\nu(s)$ the lowest eigenvalue of $T_\alpha (s)$. Using the
results of \cite[Lemma~3.6]{EK4} we conclude that
$$
\nu (s) = \xi_\alpha +\mathcal{O}(d^ {-9/2} \mathrm{e}^{C
/d})\quad \mathrm{for }\;\, s\in [0,L]\,,
$$
and
$$
\nu (s) > 0 \quad \mathrm{for }\;\, s\in [-d, 0)\cup (L, L+d]\,.
$$
Suppose that $\psi \in D(W)$ is normalized, $\|\psi
\|_{\mathcal{D}^{ex}_d}=1$.  Using (\ref{eq-defV}) together with
the above inequalities we get
\begin{eqnarray}
\lefteqn{(W\psi , \psi)_{\mathcal{D}^{ex}_d} =\left(
\Big(-\partial_s h^{-2}\partial_s -\frac14 (\gamma^\mathrm{ex})
^2\Big)\psi, \psi \right)_{\mathcal{D}^\mathrm{ex}_d}+(T_\alpha
\psi, \psi)_{\mathcal{D}^\mathrm{ex}_d}+\mathcal{O}(d)} \nonumber
\\ && \nonumber \geq \left( \Big(-\partial_s h^{-2}\partial_s
-\frac14(\gamma^{ex}) ^2\Big)\psi, \psi
\right)_{\mathcal{D}^{ex}_d}+(\nu  \psi,
\psi)_{\mathcal{D}^\mathrm{ex}_d}+\mathcal{O}(d) \\ && \geq \left(
\Big(-\partial_s h^{-2}\partial_s -\frac14(\gamma^\mathrm{ex})
^2\Big)\psi, \psi \right)_{\mathcal{D}^\mathrm{ex}_d}+\xi_\alpha
+\mathcal{O}(d) \,.
\end{eqnarray}
Using now the minimax principle in combination with the result of
\cite{EY1} we arrive at (\ref{eq-W-ev1}).

\smallskip

\noindent  \emph{Estimates for eigenvalues of $S^\mathrm{ex}_d$.}
The change-of-variable transformation
$$
s \to \frac{L}{L+2d}(s+d)
$$
turns $S_d^\mathrm{ex}$ into the operator acting in $L^2 (I)$ as
$$
\tilde{S}_d= - \left( \frac{L}{L+2d}
\right)^2\frac{\mathrm{d}^2}{\mathrm{d}s^2}-\frac14(\tilde{\gamma}_d)^2\,,
$$
where $ \tilde{\gamma}_d(s)= \gamma^\mathrm{ex}_d
\Big(\frac{L}{L+2d}(s+d)\Big)$. By construction we have
\begin{equation}\label{eq-ayxC}
\lambda_j (S^\mathrm{ex}_d) =\lambda_j  (\tilde{S}_d)\,.
\end{equation}
Moreover, since $|\tilde{\gamma}_d-\gamma | =\mathcal{O}(d)$ we
get
$$
\tilde{S}_d= -
\big(1+\mathcal{O}(d)\big)\frac{\mathrm{d}^2}{\mathrm{d}s^2}-\frac{\gamma
^2}{4}\,,
$$
which in view of (\ref{eq-ayxC}) implies
\begin{equation}\label{eq-Sd}
\lambda_j (S^\mathrm{ex}_d) =\lambda _j (S) +\mathcal{O}(d)\,.
\end{equation}
Combining this relation with (\ref{eq-W-ev1}) we arrive at the
sought lower bound (\ref{eq-lbW}).

\end{document}